\documentclass[a4paper]{article}

\usepackage[pages=all, color=black, position={current page.south}, placement=bottom, scale=1, opacity=1, vshift=5mm]{background}
\SetBgContents{
	
}      

\usepackage[margin=1in]{geometry} 

\usepackage{amsmath}
\usepackage{amsthm}
\usepackage{amssymb}
\usepackage{bbold}
\usepackage{braket}
\usepackage{array}

\usepackage{hyperref}
\usepackage{cite}
\hypersetup{
	unicode,
	pdfauthor={AmirHossein Tangestaninejad, Vahid Karimipour},
	pdftitle={Bending Hyperplanes: Nonlinear Entanglement Witnesses via Envelopes of Linear Witnesses},
	pdfsubject={Bending Hyperplanes, Nonlinear Entanglement Witnesses via Envelopes of Linear Witnesses},
	pdfkeywords={article, template, simple},
	pdfproducer={LaTeX},
	pdfcreator={pdflatex}
}

\usepackage[sort&compress,numbers,square]{natbib}
\theoremstyle{plain}
\newtheorem{theorem}{Theorem}

\theoremstyle{definition}

\usepackage{graphicx}
\graphicspath{{fig/}}
\usepackage{subcaption}

\usepackage{algorithm, algpseudocode} %
\usepackage{mathrsfs} 
\usepackage{lipsum}
\usepackage{physics}
\def\be{\begin{equation}}
	\def\ee{\end{equation}}
\def\ba{\begin{eqnarray}}
	\def\ea{\end{eqnarray}}
\def\lo{\longrightarrow}

\def\la{\langle}
\def\ra{\rangle}
\def\a{\alpha}
\def\b{\beta}

\def\L{\Lambda}
\def\D{\Delta}

\def\D{\Delta}

\def\p{{\bf p}}
\def\ni{\noindent}
\def\G{\Gamma}
\def\L{\Lambda}

\def\ni{\noindent}
\def\bex{\begin{dinglist}{110}\dsquare}
	\def\eee{\end{dinglist}}
\def\bet{\begin{dinglist}{110}\bsquare}
	\def\bfr{\begin{mdframed}[backgroundcolor=blue!20]\vspace{0.5cm}}
		\def\efr{\vspace{0.5cm}\end{mdframed}}

\title{Bending hyperplanes: Nonlinear entanglement witnesses via envelopes of linear witnesses}
\author{AmirHossein Tangestaninejad$^1$ \and Vahid Karimipour$^1$}

\date{
	$^1$\small{Deptartment of Physics, Sharif University of Technology, Tehran, Iran} \\%
}

\begin{document}
	\maketitle


\begin{abstract}
	Entanglement witnesses (EWs) are fundamental tools for detecting entanglement. However traditional linear witnesses often fail to identify most of the entangled states. In this work, we construct a family of nonlinear entanglement witnesses by taking the envelope of linear witnesses defined over continuous families of pure bipartite states with fixed Schmidt bases. This procedure effectively "bends" the hyperplanes associated with linear witnesses into curved hypersurfaces, thereby extending the region of detectable entangled states. The resulting conditions can be expressed in terms of the positive semidefiniteness of a family of matrices, whose principal minors define a hierarchy of increasingly sensitive detection criteria. We show that this construction is not limited to the transposition map and generalizes naturally to arbitrary positive but not completely positive (PnCP) maps, leading to nonlinear analogs of general entanglement witnesses. We emphasize that the required measurements remain experimentally accessible, as the nonlinear criteria are still formulated in terms of expectation values over local operator bases. Through both analytical and numerical examples, we demonstrate that the proposed nonlinear witnesses outperform their linear counterparts in detecting entangled states which may evade individual linear EWs in the construction. This approach offers a practical and conceptually elegant enhancement to entanglement detection in finite-dimensional systems.\\

\end{abstract}

\section {INTRODUCTION}
\color{black}

Detecting and characterizing entanglement remains one of the most fundamental challenges in quantum information theory \cite{NielsenChuang2000}. Among various approaches \cite{peres1996separability, horodecki1999reduction, nielsen2001separable, doherty2002distinguishing, ChenWu2003, rudolph2003cross, doherty2004complete, ioannou2004improved,hulpke2005two}, entanglement witnesses (EWs) offer a particularly practical framework due to their operational interpretation and experimental implementability \cite{Horodecki1996, Terhal2000,Barbieri2003}. A linear EW corresponds to a hyperplane in the state space that separates certain entangled states from the convex set of separable states. However, the detection power of linear witnesses is inherently limited, especially near the boundary of separability, where many entangled states remain undetected \cite{lewenstein1998separability,Guhne2002,vznidarivc2007detecting}. \\

\ni
Due to these limitations, considerable attention has recently been devoted to nonlinear entanglement witnesses \cite{hofmann2003violation,guhne2006nonlinear,guhne2007covariance,guhne2007nonlinear,moroder2008iterations,arrazola2012accessible,shen2020nonlinear,sen2023measurement,rico2024entanglement,wang2025witness}.
This line of research was initiated in Refs. \cite{guhne2006nonlinear,guhne2007nonlinear}, which introduced a simple  idea: starting from a linear entanglement witness $W$, one constructs an extended nonlinear functional $$F_{nl}(\rho):= \tr(W\rho)-\chi(\rho),$$ where $\chi(\rho)$
 is a nonlinear, typically quadratic, function of the density matrix. The function $F_{nl}(\rho)$
remains non-negative for all separable states, while becoming negative for a larger class of entangled states.
In particular, for many entangled states the nonlinear correction $\chi(\rho)$
 becomes sufficiently negative to render $F_{nl}(\rho)<0$
 even in cases where the linear term $\tr(W\rho)$
  would be positive or zero.
Geometrically, the nonlinearity in $\chi(\rho)$
 effectively bends the witness surface into a curved boundary that more tightly envelops the convex set of separable states, thereby revealing entangled states hidden beyond the flat hyperplane defined by the linear witness.
It is worth noting that $F_{nl}$
typically depends on only a few additional observables of the state, so the measurement requirements for the nonlinear witness are only marginally greater than those for its linear counterpart.\\

\ni In this work, we adopt a different approach and construct a family of nonlinear entanglement witnesses by taking the envelope of a continuous family of linear witnesses (see Fig. \ref{fig:EnvImg}).
The resulting curved surface effectively bends the hyperplanes corresponding to these linear witnesses, forming a more adaptive detection boundary that outperforms any individual linear element of the family.
This construction leads to a nonlinear matrix condition, whose positive semidefiniteness—together with the positivity of its principal minors—establishes a hierarchy of increasingly sensitive entanglement criteria.
It remains an open question whether this construction can be recast as a special case of the nonlinear framework introduced in Ref.~\cite{guhne2006nonlinear}, for an appropriate choice of the nonlinear function 
$\chi$.\\

\noindent
Unlike previous nonlinear approaches based on state-dependent optimization\cite{Guhne2003,ChenWu2003}, our method provides analytic and universal criteria derived directly from first principles.
It generalizes naturally to any positive but not completely positive (PnCP) map, yielding nonlinear analogs of canonical linear witnesses.
Importantly, all required expectation values can be expressed in terms of the same fixed set of local operators used for the linear witnesses \cite{Terhal2000,Barbieri2003}.
\noindent
We demonstrate the power of this construction through both analytical and numerical examples, including mixtures of Bell states, outputs of amplitude-damping channels, and randomly generated entangled states.
In all cases, the proposed nonlinear criteria detect a strictly larger class of entangled states than any constituent linear witness, thereby providing a practical and systematic enhancement to the theory and practice of entanglement detection.

\section{NOTATION AND CONVENTIONS}\label{section:notation}

  The complex vector space of dimension $d$ is denoted by $\mathbb{C}^d$, the space of $d$-dimensional square complex matrices by $M_d$, 
the vector space of linear operators on a Hilbert space $\mathcal{H}$ by $\mathcal{L}(\mathcal{H})$, and the convex cone of positive operators in $\mathcal{L}(\mathcal{H})$ by $\mathcal{L}^+(\mathcal{H})$ and the set of density matrices by $\mathcal{D}(\mathcal{H})$.  A Positive, but not Completely Positive (PnCP) map, is a linear Hermiticity preserving map $\Lambda: \mathcal{L}(\mathcal{H}_B)\to \mathcal{L}(\mathcal{H}_A)$ which is positive, but whose extension $\mathbb{I}\otimes \Lambda$ is no longer positive.   \\
\ni Fixing a basis $\{\ket{e_i}\}_{i=1}^{d_A}\in \mathcal{H}_A$, and $\{\ket{f_j}\}_{j=1}^{d_B}\in \mathcal{H}_B$,  the Schmidt decomposition of a pure state  $\mathcal{H}_A\otimes \mathcal{H}_B$  is written as 
\be
|\psi\ra=\sum_{i=1}^k \sqrt{p_i}|e_i,f_i\ra,
\ee
where $p_i >0$, $k$ is the Schmidt rank and the Schmidt coefficients are collected in a vector $${\bf p}=(\sqrt{p_1},\sqrt{p_2},\cdots \sqrt{p_k})^T.$$
The partial transposition is defined as a map $\mathbb{I}\otimes T:\mathcal{L}(\mathcal{H}_A \otimes\mathcal{H}_B)\to\mathcal{L}(\mathcal{H}_A\otimes\mathcal{H}_B)$ \cite{chruscinski2014entanglement}. Partial transposition of an operator $X\in \mathcal{L}(\mathcal{H}_A\otimes \mathcal{H}_B)$, where $$X=\sum_{i,j,k,l} X_{ij,kl}|e_i,f_j\ra\la e_k,f_l|$$ is given by $X^\Gamma\in \mathcal{L}(\mathcal{H}_A\otimes \mathcal{H}_B)$ where $$X^\Gamma=\sum_{i,j,k,l} X_{ij,kl}|e_i,f_l\ra\la e_k,f_j|=\sum_{i,j,k,l}^d X_{il,kj}|e_i,f_j\ra\la e_k,f_l|.$$ 
\ni We note that a linear entanglement witness (EW) is a Hermitian operator $W$ such that the negativity of the linear function  $\tr(W\rho)$ detects the entanglement of $\rho$. In contrast, a nonlinear witness does the same, except that it is a nonlinear function of the density matrix, denoted by $F(\rho)$ throughout this paper. Since we construct a hierarchy of nonlinear witnesses, we differentiate them by subscripts. This will become clear as we go ahead.\\

\ni Throughout the paper, we deal with bipartite entangled states of the form $\rho\in \mathcal{L}^+(\mathcal{H}_A\otimes \mathcal{H}_B)$ where $\mathcal{H}_B$ is isomorphic to $\mathcal{H}_A$, i.e. $\mathcal{H}_A\cong \mathcal{H}_B\cong \mathbb{C}^d$.
All the results and arguments can be readily generalized to the case where $\rho\in \mathcal{L}^+(\mathcal{H}_A\otimes \mathcal{H}_B)$, where $\text{dim}(\mathcal{H}_A)$ need not necessarily be equal to $\text{dim}(\mathcal{H}_B)$.\\

\ni Finally, we denote linear and nonlinear entanglement witnesses derived from the transposition map respectively by $W_T$ and $F_T$. The linear and nonlinear witnesses derived from a general PnCP $\Lambda$ are denoted respectively by $W_\Lambda$ and $F_\Lambda$.\\

 \section{NONLINEAR ETNAGLEMENT WITNESSES FROM TRANSPOSITION MAP}\label{section:nonlinearwitness}
 In this section, we first introduce our simplest nonlinear witnesses for any dimension and will describe their geometrical meanings in the next section. We will see that depending on the dimension of the density matrix, one can have many nonlinear witnesses and explain how our simplest witnesses are related to the transposition map and how they generalize to more general PnCP maps. We describe our results for the case where the dimensions of the two Hilbert spaces are the same, i.e. $d_A=d_B=d$. All our results can be easily generalized to the more general case, where the dimensions are different. Our first basic result is the following:
 \begin{theorem}\label{thm1}
  Fix two bases for $\mathcal{H}_A$ and $\mathcal{H}_B$ and let $\rho\in \mathcal{D}(\mathcal{H}_A\otimes \mathcal{H}_B)$ be a bipartite density matrix, $\rho=\sum_{i,j,k,l}\rho_{ij,kl}|e_i,f_j\ra\la e_k,f_l|$.   Define a matrix $\Delta_T(\rho)$ with the following entries:
  \be
  [\Delta_T(\rho)]_{ij}=\frac{1}{2}(\rho_{ij,ji}+\rho_{ji,ij}).
  \ee
  Then $\Delta_T(\rho)$ is a positive semidefinite matrix if $\rho$ is a separable state, hence negativity of $\Delta_T(\rho) $ is a nonlinear witness of the entanglement of $\rho$.\\
\end{theorem}
\begin{proof}    
	Consider a pure product state $\sigma=\ket{a}\bra{a}\otimes \ket{b}\bra{b}$, for which the elements $[\D_T(\sigma)]_{i,j}$ read
	\begin{equation}
		[\D_T(\sigma)]_{i,j} = \frac{1}{2}(a_i^* a_j b_j^* b_i + a_j^* a_i b_i^* b_j),
	\end{equation}
	where $a_i=\bra{a}e_i\rangle$ and $b_i = \bra{b}f_i\rangle$. For an arbitrary vector $\ket{v}$, we find
	\begin{align}
		\bra{v}\D_T(\sigma)\ket{v} &= \frac{1}{2}\sum_{i,j}\Big((v^*_i a^*_i b_i)(v_j a_j b^*_j) + (v^*_i a_i b^*_i)(v_j a^*_j b_j)\Big)  \geq 0.
	\end{align}
	\noindent 
	Since any separable state $\rho$ is a convex combination of pure product states, i.e.  $\rho =\sum_k p_k \ket{a_k}\bra{a_k}\otimes \ket{b_k}\bra{b_k}$, the result follows for all separable states. 
\end{proof}

\section{THE GEOMETRICAL INTERPRETATION OF NONLINEARITY}\label{section:geometric}
Although it is evident that the witness is nonlinear, a geometrical interpretation clarifies the root cause of this nonlinearity and explains why they have more detection power compared with linear witnesses, as we will see in the examples. The main idea is that we construct a nonlinear witness like Theorem \ref{thm1} as the envelope of a set of linear witnesses. In other words, we bend the set of hyperplanes defined by these linear witnesses toward the set of separable states in such a way that the set of entangled states detected by this nonlinear witness is now slightly larger than each of those linear witnesses alone.  To show this, we proceed as follows:\\

\ni Consider a family of states in Schmidt decomposed form $|\psi(\p)\ra\in \mathcal{H}_A\otimes \mathcal{H}_B,$
\be
|\psi({\bf p})\ra=\sum_{i=1}^k \sqrt{p_i} |e_i, f_i\ra.
\ee
Note that we removed the zero Schmidt coefficients. Thus, $k\leq d$ is the Schmidt rank of $\ket{\psi({\bf p})}$, and $p_i >0$.
We fix the Schmidt bases $\{|e_i\ra, |f_i\ra\}$ and only vary the parameters $\{p_i\}$ which are the squares of Schmidt coefficients $\sqrt{\{p_i\}}$. The family of linear EWs constructed from these states in the form $W(\p)=|\psi(\p)\ra\la \psi(\p)|^\Gamma$ has the explicit form
\begin{equation}\label{linearextwitness}
	W_T({\bf p})=\sum_{i,j=1}^k\sqrt{p_i p_j}\ket{e_i,f_j}\bra{e_j,f_i}.
\end{equation}
 In the space of all density matrices $\rho\in \mathcal{D}(\mathcal{H}_{AB})$, this witness defines a hyperplane with the equation    $ \tr(W({\bf p})\rho)=0$ or 
\begin{equation}\label{hyperplaneEq}
	\tr(W_T({\bf p})\rho)= \sum_{i,j}\sqrt{p_i p_j}\bra{e_j,f_i}\rho\ket{e_i,f_j}=0.
\end{equation}
Symmetrizing over the indices $i$ and $j$, we find 
\begin{equation}\label{hyperplaneEqS}
	\tr(W_T({\bf p})\rho)= \sum_{i,j=1}^k\sqrt{p_i p_j}\Delta_T(\rho)_{ij}=0,
\end{equation}
where
\be
\Delta_T(\rho)_{ij}=\frac{1}{2}(\rho_{ij,ji}+\rho_{ji,ij}).
\ee
We are interested in the collective effect of these hyperplanes as a surface on which all the hyperplanes are tangent. This is equivalent to determining the envelope \cite{bruce1992curves} of this family, which in turn is obtained by eliminating the parameter $p$ among the following set of equations \cite{bruce1992curves}. For a more detailed explanation see the Appendixes. 
\begin{equation}\label{envelopeEqs}
	\tr(W_T(\p) \rho)=0, \hspace{1.5cm} \partial_{p_n}[\tr(W(\p)\rho)]=0,\hspace{0.5cm} \forall\ n=1\cdots k,
\end{equation}
subject to the constraint $\sum_i p_i=1$. Thus, the envelope will solely depend on the elements $\rho_{ij,ji}$ of a given density matrix. To eliminate the parameters $\{{\bf p}\}$, we use the Lagrange multiplier method and define the Lagrangian function \cite{bertsekas2014constrained} as
\begin{equation}\label{lg1}
	\mathcal{L}[\boldsymbol{\rho},\p,\lambda] := \tr(W_T(\p)\rho) + \lambda(\sum_i p_i - 1).
\end{equation}
Thus, Eqs. (\ref{envelopeEqs}) become
\begin{equation}\label{lg2}
	\begin{cases}
		
		\tr(W_T(\p) \rho)=0,\\
		\\
		\partial_{p_n}  \mathcal{L}[\boldsymbol{\rho},\p,\lambda]=0,\ \ \  \forall \ n=1\cdots k,\\
		\\
		\partial_\lambda \mathcal{L}[\boldsymbol{\rho},\p,\lambda]=0.
	\end{cases}
\end{equation}
Using the explicit form (\ref{hyperplaneEq}), and the fact that $p_i$'s are nonzero to eliminate the $p_n$ from the denominator of the second equation, these equations become
\begin{equation}
	\label{EnvExtEqs}
	\begin{cases}
	\sum_{i,j=1}^k\sqrt{p_i p_j} \Delta_T(\rho)_{ij}=0,\\
	\\
		\sum_{i=1}^k \Delta_T(\rho)_{ni}\sqrt{p_i}- \lambda \sqrt{p_n}=0,\hspace{1cm} n=1...k.\\
		\\
		\sum_{i=1}^k  p_i=1.
	
	\end{cases}
\end{equation}
 In matrix form  and noting the vector ${\bf p}=(\sqrt{p_1},\sqrt{p_2},\cdots \sqrt{p_k})^T$, these equations read
\begin{equation}\label{eqset:envelopematrixform}
	\begin{cases}
			\p^{\mathbf{T}}\D_T(\rho)\p=0,\\
			\\
		\D_T(\rho) \p=\lambda \p,\\
		\\
		\p^{\mathbf{T}}\p=1.
		\end{cases}
\end{equation}
By comparing the first and last equations, it becomes evident that $\lambda=0$.
The above equations are satisfied if $\det\D_T(\rho)=0$. This very equation describes the surface that all the witnesses (\ref{hyperplaneEq}) are tangent to. Therefore the equation $\det\D_T(\rho)=0$ defines the equation of the envelope. Note that we have already shown in Theorem \ref{thm1} that the matrix $\D_T(\rho)$ is positive semidefinite for all separable states. 
\begin{figure}[H]
	\centering
	\includegraphics[width=14cm]{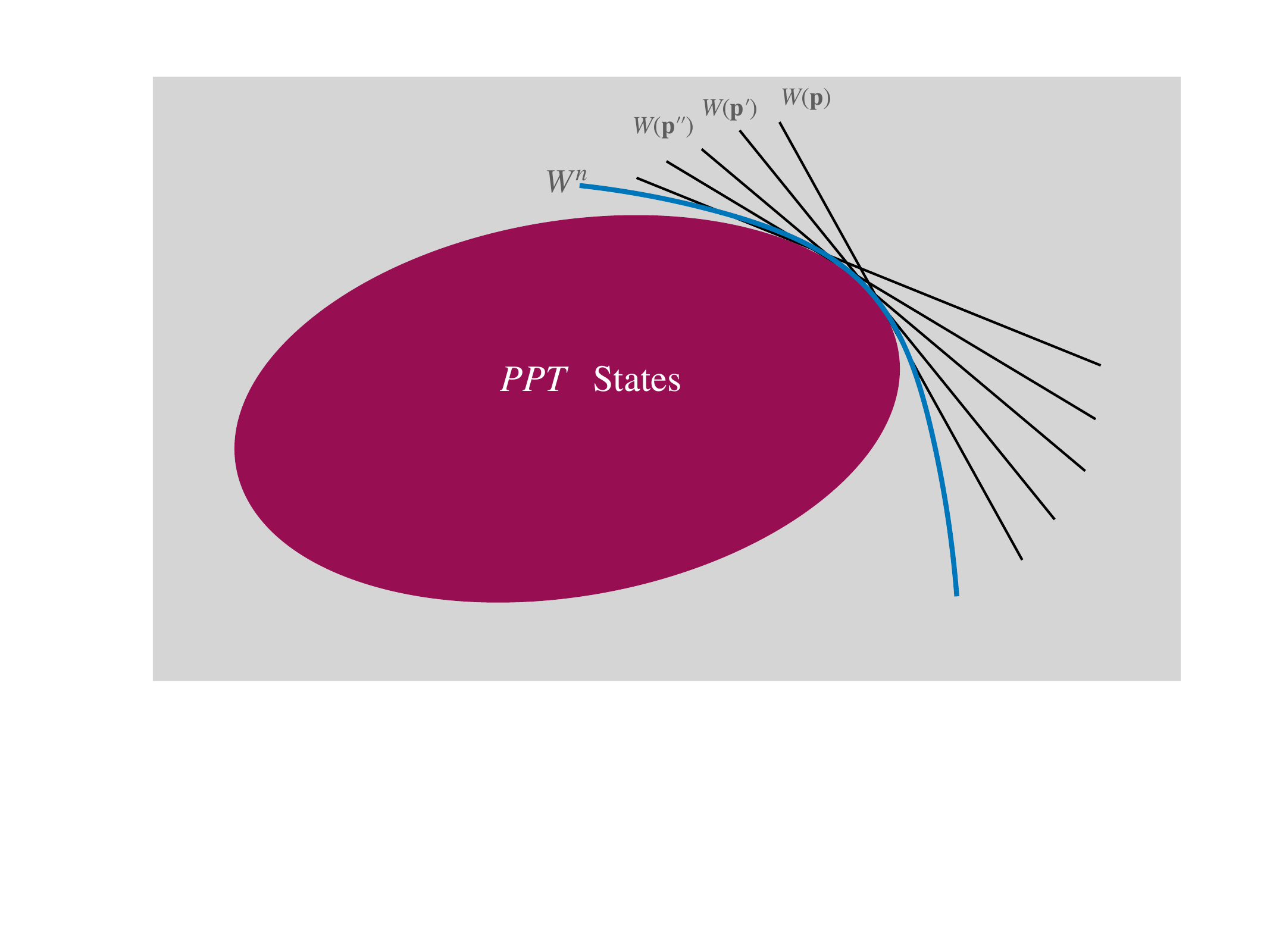}\vspace{-2.5cm}
	\caption{Schematic image of the set of bipartite quantum states $\mathcal{D}(\mathcal{H}_{AB})$, set of separable states, and the set of PPT (positive partial transpose) entangled states. The lines represent some witnesses belonging to the family $W_T(\p)$, parametrized by the set of positive coefficients $\{p_i\}$, given via the equation $\tr(W_T(\p) \rho)=0$. The envelope of this family yields the equation $F_T(\rho)\equiv\det\D_T(\rho)=0$, which is the surface that is tangent to all the hyperplanes  $\tr(W_T(\p)\rho)=0$. }
	\label{fig:EnvImg}
\end{figure}

\ni Therefore we find that if $\Delta_T(\rho)<0$ then $\rho$ should be entangled. This is a nonlinear entanglement witness that is constructed from the envelope of hyperplanes constructed from the linear EW's and is more powerful than each of these linear witnesses. Moreover, since in each family, the optimization has been done only on the coefficients of the Schmidt coefficients and not the Schmidt bases, this nonlinear witness requires the same type of local measurements as the linear EW's.
It is worth noting the explicit form of the matrix  $\Delta_T(\rho)$, whose determinant, being negative, certifies the entanglement of the state.

\subsection{The effect of change of basis} Suppose that we choose a different basis and calculate  $F_T'(\rho)=\det \Delta_T'(\rho)$, where $\Delta'_T(\rho)_{i,j}:=\Delta_T(\rho)_{i',j'}$, in which the bases of $\mathcal{H}_A$ and $\mathcal{H}_B$ are related to the old ones as $|e'_i\ra=U|e_i\ra$ and $|f'_j\ra=V|f_j\ra$, where $U$ and $V$ are unitary operators. Then one finds that $\Delta'_T(\rho)=\Delta_T ((U^\dagger\otimes V^\dagger)\rho (U\otimes V))$ and hence $F_T'(\rho)=F_T((U^\dagger\otimes V^\dagger)\rho (U\otimes V))$. This means that while $F$ detects the entanglement of a set of states $S(F_T):=\{\rho| F_T(\rho)<0\}$, $F'$ detects the entanglement of the locally equivalent states $S(F'_T)=(U^\dagger\otimes V^\dagger)S(F_T)(U\otimes V)$. \\

\section{Nonlinear witnesses from general positive maps}\label{section:generalization}
The nonlinear EW's that we proposed in previous sections are derived from the linear ones corresponding to transposition maps. 
 We now generalize these results to any arbitrary entanglement witness constructed from a general PnCP map $\Lambda$. It is well known that there is a one-to-one correspondence between PnCP maps and entanglement witnesses, that is,  any entanglement witness $W$ which detects entanglement of states in $\mathcal{L}^+(\mathcal{H}_A\otimes \mathcal{H}_B)$ is nothing but the Choi matrix of a PnCP $\Lambda$. An alternative way for this correspondence is that a witness can be written as \cite{guhne2009entanglement} 
 \begin{equation}\label{generalfamilywitness}
	W_\L=\mathbb{I}\otimes\Lambda^\dagger\ket{\psi}\bra{\psi},
\end{equation}
for $\ket{\psi}\in \mathcal{H}_A\otimes\mathcal{H}_A$. We adopt this alternative correspondence and, as in Sec. \ref{section:geometric}, we fix the map $\Lambda^\dagger:\mathcal{L}(\mathcal{H}_A)\to\mathcal{L}(\mathcal{H}_B)$ \cite{Horodecki1996,chruscinski2014entanglement} and vary the parameters  $p_i$ of $|\psi(\p)\ra$
\be
|\psi(\p)\ra=\sum_{i=1}^k \sqrt{p_i}|e_i\ra_A |e'_i\ra_A,
\ee
while keeping the Schmidt bases $|e_i,e'_i\ra$ fixed. Therefore, similar to the case in Sec. \ref{section:geometric},  we denote the entanglement witness as 
\begin{equation}\label{generalfamilywitnessexplicit}
	W_\L(\p)=\sum_{i,j=1}^k \sqrt{p_ip_j}|e_i\ra\la e_j|\otimes \Lambda^\dagger(|e'_i\ra\la e'_j|),
\end{equation}
The equation $\tr(W_\L(\p)\rho)=0$ defines a hyperplane with equation
\begin{equation}\label{eq:trWgeneral}
	\tr(W_\L(\p)\rho)=\sum_{i,j}\sum_{l,m}\sqrt{p_i p_j}\rho_{jl,im}\Lambda^{\dagger}_{ml,ij}=0,
\end{equation}
where 
$$(\Lambda^\dagger)_{ml,ij}:= \bra{f_m}\Lambda^\dagger(\ket{e'_i}\bra{e'_j})\ket{f_l}.$$
Symmetrizing the right-hand side of (\ref{eq:trWgeneral}) over $i$ and $j$ and defining 
\begin{equation}\label{DeltaijDualmap}
	(\Delta_\L)_{i,j}(\rho):=\frac{1}{2}\sum_{l,m} \Big(\rho_{jl,im} (\Lambda^\dagger)_{ml,ij}+\rho_{il,jm}(\Lambda^\dagger)_{ml,ji}\Big)=\sum_{l,m}\Re(\rho_{il,jm} (\Lambda^\dagger)_{ml,ji}).
\end{equation}
we find for the equation of the hyperplane
\be
	\tr(W_\Lambda({\bf p})\rho)=\sum_{i,j}\sqrt{p_ip_j}({\Delta_\Lambda})_{ij}(\rho)=0.
	\ee
Finding the envelope of all these hyperplanes when we vary the parameters $p_i$, we find the same equations as in (\ref{eqset:envelopematrixform}), with $\Delta_T$ replaced with $\Delta_\Lambda$, i.e.
\begin{equation}
	\begin{cases}
		\p^{\mathbf{T}}\D_\Lambda(\rho)\p=0,\\
		\\
			\D_\Lambda(\rho) \p=\lambda \p,\\
		\\
			\p^{\mathbf{T}}\p=1,
	\end{cases}
\end{equation}
leading to the nonlinear witness
\begin{equation}\label{generalNLEW}
	F_\Lambda(\rho)=\det \Delta_\L(\rho).
\end{equation}
Negativity of this quantity is a witness of entanglement of the state $\rho$. It remains to show that $\Delta_\L(\rho)$ is non-negative for all separable states. This can be done similarly to what we did for the transpose map $\Delta_T$ in Sec. \ref{section:nonlinearwitness}. We have
\begin{theorem}
	  Let $\rho\in \mathcal{D}(\mathcal{H}_A\otimes \mathcal{H}_B)$
      be bipartite density matrix, and let $\Lambda:\mathcal{L}(\mathcal{H}_B)\lo \mathcal{L}(\mathcal{H}_A)$ be a PnCP map. The basis states of the two spaces are the same as in Theorem \ref{thm1}. Define a matrix $\Delta_\L(\rho)\in M_d$ with the following entries:
	$
	\Delta_\L(\rho)_{ij}
	$
	defined in (\ref{DeltaijDualmap}), is semidefinite positive, for any separable state $\rho$. 
	 Therefore negativity of $\Delta_\L(\rho) $ is a nonlinear witness of the entanglement of $\rho$.\\
\end{theorem}
\begin{proof}
	The same argument as in the proof of Theorem \ref{thm1},  applies here, with minor adjustment. Namely, we take a separable product state $\sigma=|a\ra\la a|\otimes |b\ra\la b|$ and find from (\ref{DeltaijDualmap})
  \begin{equation}
	\Delta_\L(\sigma)_{i,j} = \frac{1}{2}\sum_{l,m} \Big (a^*_j a_i b^*_l b_m \bra{f_m}\Lambda^\dagger(\ket{e'_i}\bra{e'_j})\ket{e_l} + a^*_i a_j b^*_l b_m \bra{f_m}\Lambda^\dagger(\ket{e'_j}\bra{e'_i})\ket{f_l}\Big ),
\end{equation}
or by defining $|x_i\ra:=a_i |e'_i\ra$ and $|b\ra:=\sum_l b^*_l|f_l\ra$,
  \begin{equation}
	\Delta_\L(\sigma)_{i,j} = \frac{1}{2} \Big(\bra{b}\Lambda^\dagger(\ket{x_i}\bra{x_j})\ket{b} + \bra{b}\Lambda^\dagger(\ket{x_j}\bra{x_i})\ket{b}\Big).
\end{equation}
Therefore for any vector $|v\ra$, we find
\begin{align}
	\bra{v}\D_\Lambda(\sigma)\ket{v} &= \frac{1}{2} \Big(\bra{b}\Lambda^\dagger(\ket{x_v}\bra{x_v})\ket{b} + \bra{b}\Lambda^\dagger(\ket{x_v}\bra{x_v})\ket{b}\Big)\geq 0,
\end{align}
where $|x_v\ra:=\sum_i v_i |x_i\ra$. 
\end{proof}
\ni The same argument as in Sec. \ref{section:nonlinearwitness}, in the case of witnesses $\Delta_T$
, are valid here. That is if $\Delta_\L(\rho)<0$, this means that  the state $\rho$ is entangled. This leads to a whole set of nonlinear witnesses $F^S_{[ij]}:=\det(\Delta_\L([\rho]^S_{ij}))$, where $[\rho]^S_{ij}$ is the principal submatrix of $\rho$ pertaining to the square block $(i,j)$.

\section{EXAMPLES: QUBIT STATES}\label{section:examples:qubit}
In this section, we apply the nonlinear witnesses defined in previous sections to a few families of bi-partite density matrices for two qubits.

\subsection{Nonlinearizing the family of witnesses that require minimal set of local measurements}
In  \cite{riccardi2020optimal}, it was shown that two-qubit optimal EWs, of the form $W=|\psi\ra\la \psi|^\Gamma$ which require a limited set of measurements (i.e. measurements of the form $\sigma_a\otimes \sigma_a, \  a=x,y,z$) are restricted to the following six families: 

\begin{equation}\label{eq:allEquations}
	\begin{aligned}
		\ket{\psi_{1}} &= a\ket{\phi^{+}}+b\ket{\phi^{-}}; \quad & \ket{\psi_{2}} &= a\ket{\psi^{+}}+b\ket{\psi^{-}};  \\
		\ket{\psi_{3}} &= a\ket{\phi^{+}}+b\ket{\psi^{+}}; \quad & \ket{\psi_{4}} &= a\ket{\phi^{-}}+b\ket{\psi^{-}}; \\
		\ket{\psi_{5}} &= a\ket{\phi^{+}}+ib\ket{\psi^{-}}; \quad & \ket{\psi_{6}} &= a\ket{\phi^{-}}+ib\ket{\psi^{+}},
	\end{aligned}
\end{equation}

One can readily see that this property is due to the nature of the Schmidt decomposition of these six families of states. In fact, a simple calculation shows that 
these states have the following Schmidt decomposition
\begin{equation}\label{eq:allEquations1}
	\begin{aligned}
		\ket{\psi_{1}} &= \a\ket{z_+,z_+}+\b\ket{z_-,z_-}; \quad & \ket{\psi_{2}} &= \a\ket{z_+,z_-}+\b\ket{z_-,z_+};  \\
		\ket{\psi_{3}} &= \a\ket{x_+,x_+}+\b\ket{x_-,x_-}; \quad & \ket{\psi_{4}} &= \a\ket{x_-,x_+}+\b\ket{x_+,x_-}; \\
		\ket{\psi_{5}} &= \a\ket{y_-,y_+}+\b\ket{y_+,y_-}; \quad & \ket{\psi_{6}} &= \a\ket{y_+,y_+}+\b\ket{y_-,y_-},
	\end{aligned}
\end{equation}
where $\a=\frac{a+b}{\sqrt{2}}$,  $\beta=\frac{a-b}{\sqrt{2}}$,  $|z_\pm\ra=\{|0\ra,|1\ra\}$ are eigenstates of the $\sigma_z$ operator, $|x_\pm\ra=\frac{1}{\sqrt{2}}\{|0\ra\pm |1\ra\}$ are eigenstates of the $\sigma_x$ operator,and $|y_\pm\ra=\frac{1}{\sqrt{2}}\{|0\ra\pm i|1\ra\}$ are eigenstates of the $\sigma_y$ operator.
This explains why witnesses of the form $W_i\equiv |\psi_i\ra\la \psi_i|^\Gamma$ require only a limited set of measurements of diagonal operators $\sigma_a\otimes\sigma_a,\ a=x,y,z$. Following the construction which led from the family of linear witnesses (\ref{linearextwitness}) to the nonlinear witnesses in Theorem \ref{thm1}, shows that the nonlinearization of each of these witnesses respectively leads to the corresponding nonlinear witness. For example the first  family of linear witness related to $|\psi_1\ra=\sqrt{p_1}|0,0\ra+\sqrt{p_2} |1,1\ra$ corresponding to the Schmidt basis
\be
|e_0\ra=|0\ra, \ \  |e_1\ra=|1\ra, \ \  |f_0\ra=|0\ra, \ \  |f_1\ra=|1\ra, \ \ 
\ee
lead to the following: 
\be
\Delta^1_T(\rho)=\begin{pmatrix}\rho_{00,00}&\Re\rho_{01,10}\\ \Re\rho_{10,01}&\rho_{11,11}\end{pmatrix}
\ee
and  
\be\label{eq:F1 T}
F_T^{(1)}(\rho) = \rho_{00,00}\rho_{11,11}-\Re(\rho_{01,10})^2.   
\ee

On the other hand,  the second  family of linear witness related to $|\psi_2\ra=\sqrt{p_1}|0,1\ra+\sqrt{p_2} |1,0\ra$ corresponding to the Schmidt basis
\be
|e_0\ra=|0\ra, \ \  |e_1\ra=|1\ra, \ \  |f_0\ra=|1\ra, \ \  |f_1\ra=|0\ra, \ \ 
\ee
lead to the following: 
\be
\Delta^2_T(\rho)=\begin{pmatrix}\rho_{01,01}&\Re(\rho_{00,11})\\ \Re(\rho_{11,00})&\rho_{10,10}\end{pmatrix}
\ee

and  

\be\label{eq: F2 T}
F_T^{(2)}(\rho) = \rho_{01,01}\rho_{10,10}-\Re(\rho_{00,11})^2.   
\ee
\color{black}

\ni 
When we are dealing with a two-qubit state, i.e. when $\text{dim}(\mathcal{H}_A)=\text{dim}(\mathcal{H}_B)=2$, the nonlinear EW's are given by 
\begin{equation}
	F^{(i)}_T = \det \Delta^{i}_T(\rho),
\end{equation}
for a $2\times 2$ matrix $\Delta^i_T(\rho)$.
These nonlinear witnesses require measurement of matrix elements of the density matrix in a certain basis, and their measurement is the same as the one demanded by the corresponding linear witness. Only the final statistical analysis is different. For example for the first two witnesses in (\ref{eq:allEquations}), we require the following elements and measurements: 
\ba
\rho_{00,00}&=&\frac{1}{4}\tr((\mathbb{I}+\sigma_z)\otimes (\mathbb{I}+\sigma_z)\rho)\cr
\rho_{11,11}&=&\frac{1}{4}\tr((\mathbb{I}-\sigma_z)\otimes (\mathbb{I}-\sigma_z)\rho)\cr
\Re(\rho_{01,10})&=&\frac{1}{2}\tr((\sigma_x\otimes\sigma_x+\sigma_y\otimes \sigma_y)\rho),
\ea
all of them requiring measurement of diagonal Pauli operators $\sigma_a\otimes \sigma_a,\ \ a=x,y,z$. 	Note that measurement of $\mathbb{I}\otimes \sigma_a$ or $\sigma_a\otimes \mathbb{I}\ \ \ a\in \{x, y, z\}$, is performed by measurement of $\sigma_a\otimes \sigma_a$ and neglecting the statistics obtained by one side. In this sense, only a limited set of measurements is required for this witness.\\

\ni When $d>2$, we note that the positivity of $\Delta_T(\rho)$ implies the positivity of all its minors, which is the determinant of all its principal submatrices \cite{johnson1990matrix} (those obtained by removing the same rows and columns from the original matrix). Therefore if any of these minors are negative, it is a witness of the entanglement of the state. The simplest minors are the two-dimensional ones. Therefore for any pair of indices $[i,j]\in \{1,\cdots k\}$, we have a witness ${F^{[ij]}_T}$ defined as 
\be
{F^{[ij]}_T}(\rho)=\rho_{ii,ii}\rho_{jj,jj}-\Re(\rho_{ij,ji})^2.
\ee
\ni Measuring any of the witnesses $F^{[ij]}_T$ (or those corresponding to larger minors) requires measurements of the same diagonal set of Pauli operators, embedded in the corresponding block $(i,j)$, i.e. measurements of the Gell-Mann matrices with support in the block $(i,j)$, $X_{ij}\otimes X_{ij}$, $Y_{ij}\otimes Y_{ij}$ and $Z_{ij}\otimes Z_{ij}$, where
\be
X_{ij}=|i\ra\la j|+|j\ra\la i|, \    Y_{ij}=-i(|i\ra\la j|-|j\ra\la i|), \   Z_{ij}=|i\ra\la i |-|j\ra\la j|. 
\ee
\ni 
The advantage of measurement of different entanglement witnesses ${F^{[ij]}_T}$ is that one can detect which part of the state, i.e. which basis states, are responsible for entanglement of the whole state $\rho$. Note that one can measure larger minors which in turn determine if a larger part of the state (a larger set of basis states) is responsible for the entanglement of the whole state.  \\

\subsection{ Convex combination of Bell states}\label{example: Bell states}
It is well known that no linear entanglement witness can detect the entanglement of more than one Bell state \cite{ioannou2006quantum}. The proof in \cite{ioannou2006quantum} is simple and beautiful: If there exists an entanglement witness $W$ such that for two Bell states $|\phi_1\ra$, and $|\phi_2\ra$, $\la \phi_{1,2}|W_T|\phi_{1,2}\ra=-\epsilon_{1,2}<0$, then one can construct two separable pure states of the form $|\pm\ra=\frac{1}{\sqrt{2}}(|\phi_1\ra\pm |\phi_2\ra)$ such that $\la +|W_T|+\ra$ and $\la -|W_T|-\ra$ cannot be both positive, leading to a contradiction. Intuitively, this theorem can be understood by noting the fact that 
the mixed state $\rho=\frac{1}{2}(|\phi_1\ra\la \phi_1|+|\phi_2\ra\la \phi_2|)$ is always separable. In fact, as shown in Fig. \ref{fig: Bell}, a line representing a convex combination of these two Bell states always touches the set of separable states. 
\begin{figure}[H]
	\centering
	\includegraphics[width=15cm]{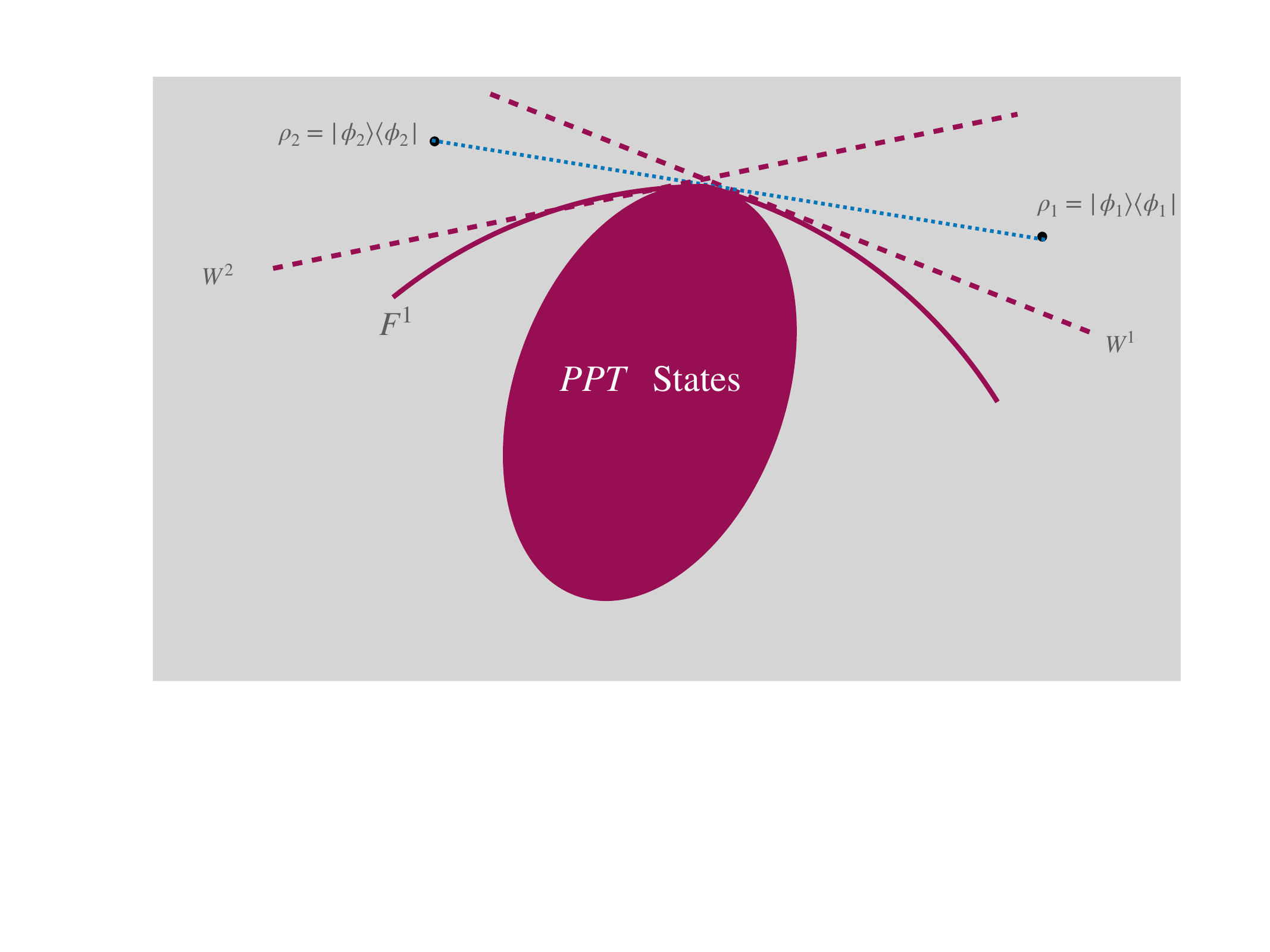}\vspace{-3cm}
	\caption{The linear EW, $W_1$ can detect the entanglement of the Bell state $|\phi_1\ra$, but fails to detect that of $|\phi_2\ra$. Similarly, the linear EW, $W_2$ can detect the entanglement of the Bell state $|\phi_2\ra$, but fails to detect that of $|\phi_1\ra$. A nonlinear EW $F_T$ can detect the entanglement of both $|\phi_1\ra$ and $|\phi_2\ra$ and their convex combinations. }
	\label{fig: Bell}
\end{figure} 

\ni For this reason, as seen in this figure, any linear EW which detects one of them, fails to detect the other. In fact, any linear witness can detect the entanglement of a convex combination of two Bell states, $\rho = x|\phi_1\ra\la \phi_1|+(1-x)|\phi_2\ra\la \phi_2|$, only for part of the interval of the parameter $x$. This is clearly depicted in Fig. \ref{fig: Bell}.\\

\ni As another example of a convex combination of two Bell states, consider the  state
\be
\rho=x|\psi^+\ra\la \psi^+|+(1-x)|\psi^-\ra\la \psi^-|=\frac{1}{2}\begin{pmatrix}0&0&0&0\\ 0 &1 &2x-1&0 \\ 0 & 2x-1&1&0\\ 0&0&0&0\end{pmatrix}.
\ee
For a linear EW like $W_T^+=|\phi^+\ra\la \phi^+|^\Gamma$, we have 
\be
\tr(W_T^+\rho)=(x-\frac{1}{2})
\ee
which shows that $W_T^+$ can detect the entanglement of $\rho$ only for $x\in[0,\frac{1}{2})$. On the other hand, for a witness $W_T^-=|\phi^-\ra\la \phi^-|^\G$, we have 
\be
\tr(W_T^-\rho)=(\frac{1}{2}-x)
\ee
which shows that $W^-$ can detect the entanglement of $\rho$ only for $x\in (\frac{1}{2},1]$.
However, a nonlinear EW can detect the entanglement of this state $\rho$ for the whole parameter range of $x\in[0,1]$. Intuitively this is done by bending the two hyperplanes, corresponding to the linear witnesses $W_1$ and $W_2$, and making them into a curved surface,  as shown in Fig. \ref{fig: Bell}. In explicit matrix form, we have 
\be
F^{(1)}_T(\rho)= \rho_{00,00}\rho_{11,11}-\Re(\rho_{01,10})^2=-\frac{1}{4}(2x-1)^2
\ee
which clearly shows its superiority over the two linear witnesses. \\

\ni The question arises whether or not the convex combination of other Bell states can also be detected by a suitable nonlinear witness. The answer is in fact positive. 
From what we learnt on change of the effect of change of basis, and the fact that all Bell states are transformed to each other by local unitary transformations, we find these states are also detected by suitable nonlinear EW's by a simple change of basis. As an example, consider 
$$\sigma=x|\phi^+\ra\la \phi^+|+(1-x)|\phi^-\ra\la \phi^-|=\frac{1}{2}\begin{pmatrix}1&0&0&2x-1\\ 0 &0 &0&0 \\ 0 & 0&0&0\\ 2x-1&0&0&1\end{pmatrix}.$$ 
Since $|\phi^\pm\ra=(I\otimes \sigma_x)|\psi^\pm\ra$, we find that $$F_T^{(2)}(\sigma)=\sigma_{01,01}\sigma_{10,10}-\Re(\sigma_{00,11})^2=-\frac{1}{4}(2x-1)^2,$$
again showing that this state is entangled for all values of $x$ except for $x=\frac{1}{2}$. A similar argument also applies to other combinations of Bell states.\\

\subsection{ The output of an  amplitude channel}
It is now instructive to consider the power of linear EW's with a nonlinear EW for an interesting example, studied in Ref. \cite{riccardi2020optimal}. Consider the amplitude-damping channel when acting on one qubit of the state $|\phi^+\ra$, and producing the state $\rho=(\mathbb{I}\otimes {\cal E})|\phi^+\ra\la \phi^+|$, where ${\cal E}(\rho)=A_0\rho A_0^\dagger + A_1\rho A_1^\dagger$, with $A_0=\begin{pmatrix}1&0\\ 0 &\sqrt{1-\gamma}\end{pmatrix}$ and $A_1=\begin{pmatrix}0&\sqrt{\gamma}\\ 0 &0\end{pmatrix}$ and $0\leq \gamma\leq 1$. One finds that
\be
\rho=\frac{1}{2}\begin{pmatrix}1&0&0&\sqrt{1-\gamma}\\ 0 &0&0&0\\ 0 &0&\gamma&0\\ \sqrt{1-\gamma}&0&0&1-\gamma \end{pmatrix}.
\ee
Peres's criterion \cite{peres1996separability} shows that this state is entangled for $0\leq \gamma<1$, that is, for the whole range of the parameter $\gamma$, except at $\gamma=1$, when the amplitude-damping channel acts as the identity channel. In Ref. \cite{riccardi2020optimal}, it has been shown that the linear witnesses for the form $W=|\psi\ra\la\psi|^\Gamma$, where $|\psi\ra=a|\psi^+\ra+b|\psi^-\ra$ (with $a^2+b^2=1$) can only detect entanglement of this state provided that one chooses the range of the parameter $a$ appropriately. For example, it has been shown if $\gamma = 0.9$, then $\tr(W\rho)<0, $ only if $0.17<a< \frac{1}{\sqrt{2}}$, while if $\gamma=0.95$,  then $\tr(W^2\rho)<0, $ only if $0.38<a< \frac{1}{\sqrt{2}}$. However, we can detect the entanglement of this state for all values of $\gamma$ by using one single nonlinear EW, namely $F^{(2)}(\rho)$, which is the nonlinearized form of the witness $W_2=\ket{\phi_2}\bra{\phi_2}^\G$ with $\ket{\phi_2}=\sqrt{p_1}|01\ra+\sqrt{p_2}|10\ra$.  This is seen simply by noting that 
\be
F^{(2)}_T(\rho)=\rho_{01,01}\rho_{10,10}-(\Re{\rho_{00,11}})^2=\frac{1}{4}(\gamma-1),
\ee
which is always negative when $\gamma<1$.  

\subsection{ Randomly generated entangled states }\label{section:Numcompare}
In this section we compare the detection power of an optimal linear witness of the form $W = \ket{\psi_1}\bra{\psi_1}^\Gamma$, with $\ket{\psi_1}$ given by (\ref{eq:allEquations})  
	$
	\ket{\psi_1} = a \ket{\phi^+} + b \ket{\phi^-},
	$  with its nonlinear counterpart. 
	To this end, we start by 
	randomly selecting complex square matrices $X$ from a Ginibre ensemble \cite{zyczkowski2011generating,johansson2012qutip}, and construct random matrices  as $\rho_r := \frac{XX^\dagger}{\tr(XX^\dagger)}$.  The measure induced by the Ginibre ensemble, coincides with the Hilbert-Schmidt measure on density matrices \cite{sommers2004statistical}. In order to select only entangled states, we apply the criterion of positive partial transpose \cite{peres1996separability} (which is necessary and sufficient in this dimension \cite{Horodecki1996}) and filter out separable states. For each value of the parameter $a$, the fraction of states detected by the linear and nonlinear entangled witnesses is determined and plotted in Fig. \ref{fig:comp}.

\begin{figure}[H]
	\centering
	\includegraphics[width=11cm]{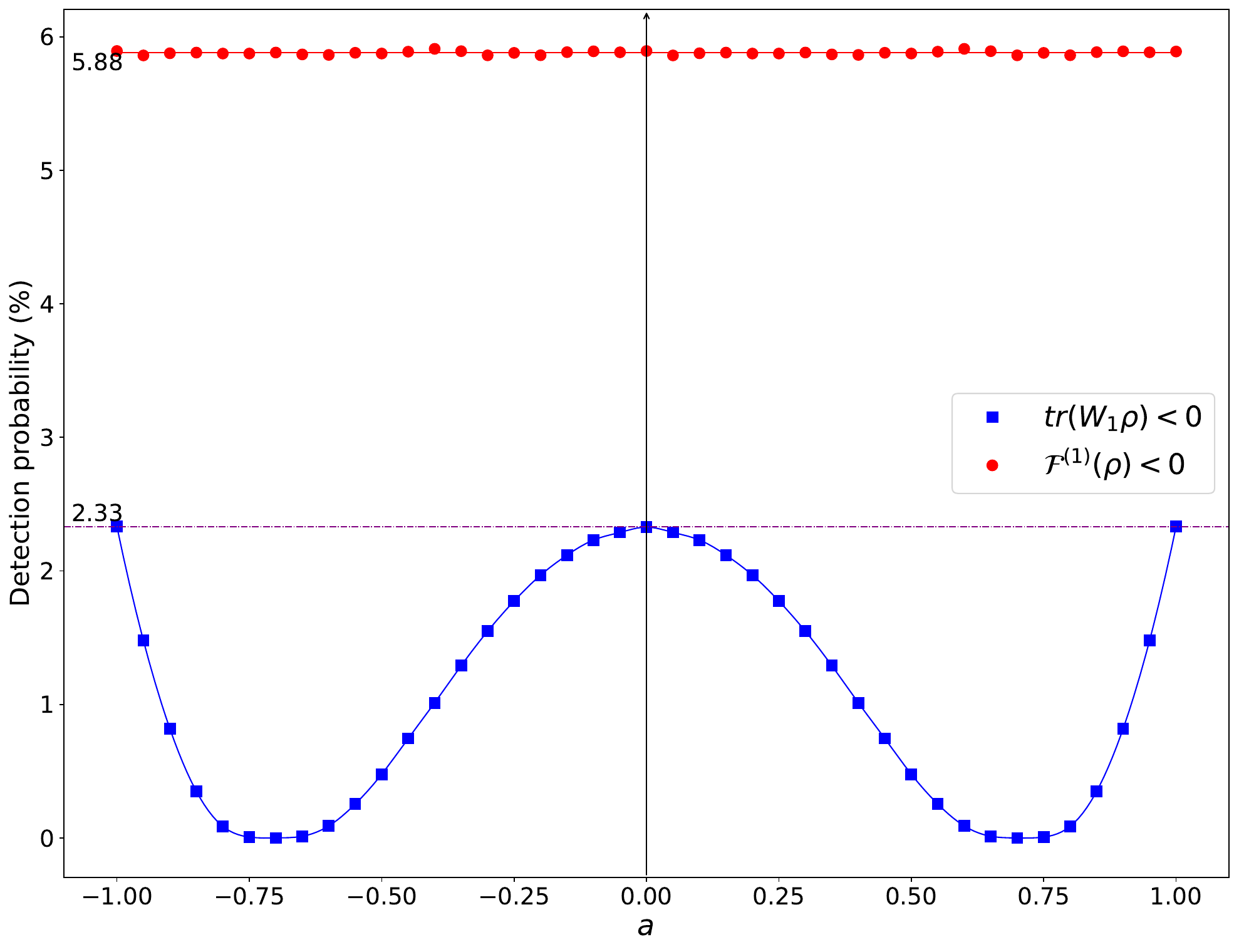}
	\caption{A comparison of detection power of entanglement witness $W_1 = \ket{\psi_1}\bra{\psi}^\Gamma$ with $\ket{\psi_1}= a\ket{\phi^+}+b\ket{\phi^-}$ (\ref{eq:allEquations}) and its nonlinear counterpart (\ref{eq:F1 T}). For each value of the parameter $a$, $6\times 10^6$ random entangled states have been examined.}
	\label{fig:comp}
\end{figure}

\section{EXAMPLES: QUTRIT STATES}\label{section:examples:qutrit}
In this section, we study two examples to show how our nonlinear witnesses are more capable of detecting entangled states than linear ones. The first one is related to the transposition map and the second to the generalized Choi map. In these two examples, the nonlinear witness $F$ takes the form of the determinant of a  $3\times 3$ matrix. 
\subsection{The nonlinear witness  derived from the transposition map}
\noindent Consider now the linear EW defined as 
$W_T=|\phi_3^+\ra\la \phi_3^+|^\Gamma$, where $|\phi_3^+\ra=\frac{1}{\sqrt{3}}(|00\ra+|11\ra+|22\ra)$ is the maximally entangled state. 
We compare the detection power of this linear EW with the nonlinear EW derived from the transposition map 
which is
\begin{equation}\label{NonlineartwoqutritNPTmaxent}
    F_T(\rho)=\det \D_T(\rho)_{3\times 3}=\det \begin{pmatrix}
        \rho_{00,00} & \Re(\rho_{01,10}) & \Re(\rho_{02,20})\\
        \Re(\rho_{01,10}) & \rho_{11,11} & \Re(\rho_{12,21})\\
        \Re(\rho_{02,20}) & \Re(\rho_{12,21}) & \rho_{22,22}
    \end{pmatrix}.
\end{equation}
To do this comparison, we use two different ensembles of density matrices, namely the Hilbert-Schmidt ensemble and the Bures ensemble. The first ensemble is produced by  
$\rho_r=\frac{HH^\dagger}{\tr(HH^\dagger)}$, where the matrix $H$ is a matrix drawn from the Ginibre ensemble \cite{bruzda2009random} and the second is produced by $\frac{(\mathbb{I}+U)HH^\dagger(\mathbb{I}+U)^\dagger}{\tr[(\mathbb{I}+U)HH^\dagger(\mathbb{I}+U)^\dagger]}$
 \cite{zyczkowski2011generating,johansson2012qutip}, where $U$ is a random Haar unitary matrix and $H$ is again drawn from a Ginibre ensemble \cite{al2010random}.
Table \ref{tab:comparetwoqutrit} compares the detection power of these two witnesses for the aforementioned ensembles. 
\begin{table}[h!]
\centering
\begin{tabular}{|c|c|c|c|}
\hline
Criterion to detect NPT states & $\tr(W_T\rho)<0$ & $F_T(\rho)<0$  \\
\hline
     Hilbert-Schmidt ensemble    &    0.12      &  1.47      \\
\hline
    Bures ensemble &  0.65 & 6.56 \\
\hline
\end{tabular}
\caption{The numbers in the last two columns show the percentage of states detected by the linear and nonlinear EWs of our $ 10^7$ two-qutrit \textit{NPT} random density matrices.}
\label{tab:comparetwoqutrit}
\end{table}

\noindent While the percentage of the states detected by both witnesses is small,  it is clear that the nonlinear witness has a significantly higher detection power compared with the linear EW.  The improvement is much larger than that of the two qubits, since the nonlinearity [the order of the polynomial derived from the determinant of $\Delta_T(\rho)$)] increases with dimension.

\color{black}
\subsection{ The  nonlinear witness derived from the generalized Choi map }
In this section,  we consider a different PnCP, construct its corresponding nonlinear EW (\ref{generalNLEW}), and compare its detection power with the corresponding linear EW.\\

\ni {\bf The generalized Choi map:} Consider a nondecomposable map on qutrit states. Choi provided the first such example \cite{choi1975completely}, which was later unified with the reduction map \cite{horodecki1999reduction,cerf1999reduction} for qutrits as the generalized Choi map $\Phi[a,b,c]$. Thus the generalized Choi map is parametrized by three real parameters $a,b$, and $c$ in Ref. \cite{chruscinski2011geometry}. Consider this generalized map whose action on an matrix $X\in M_3(\mathbb{C})$ is given by 
\begin{multline}\nonumber\label{3parametermap}
    \Phi[a,b,c]\begin{pmatrix}
        X_{0,0} & X_{0,1} & X_{0,2}\\
        X_{1,0} & X_{1,1} & X_{1,2}\\
        X_{2,0} & X_{2,1} & X_{2,2}
    \end{pmatrix}
    \\ =\mathcal{N}_{a,b,c}\begin{pmatrix}
        a X_{0,0}+b X_{1,1} + c X_{2,2} & -X_{0,1} & - X_{0,2}\\
        -X_{1,0} & c X_{0,0} + a X_{1,1} + b X_{2,2} & -X_{1,2} \\
        - X_{2,0} & -X_{2,1} & b X_{0,0} + c X_{1,1} + a X_{2,2}
      \end{pmatrix}.
\end{multline}
Here  $\mathcal{N}_{a,b,c}=\frac{1}{a+b+c}$ is a normalization factor which we will neglect hereafter since it plays no role in the construction of witnesses. The map $\Phi[a,b,c]$ is a PnCP if and only if 
\begin{enumerate}
\item  \( 0\leq a <2\),
\item  \( a+b+c\geq 2\),
\item  if \( a\leq 1\),  then \(bc\geq (1-a)^2\).
\end{enumerate}

 For $[a,b,c]=[0,1,1]$, this map reproduces the reduction map, i.e. 
 \begin{equation}
 	\Phi_R(X) = \mathbb{I}_d \tr(X)-X.
 \end{equation}
 For $[a,b,c]=[1,1,0]$ and $[1,0,1]$, it reduces to the Choi map and its dual respectively. Moreover, for $bc<\frac{(2-a)^2}{4}$ this map is indecomposable, and for $a+b+c=2$, these maps are optimal \cite{ha2011one,chruscinski2013optimal}. The boundary of this domain is parametrized by one single parameter $\theta$\cite{chruscinski2011geometry}, where $a$, $b$, and $c$ are given by 
\begin{equation}\label{parametrization}
	\begin{aligned}
		a(\theta) &= \frac{2}{3}(1+\cos (\theta)),\\
		b(\theta) &= \frac{2}{3}(1-\frac{\cos(\theta)}{2} - \frac{\sin (\theta)\sqrt{3} }{2}),\\
		c(\theta) &= \frac{2}{3}(1-\frac{\cos(\theta)}{2} + \frac{\sin (\theta)\sqrt{3} }{2}).
	\end{aligned}
\end{equation}

 \ni For a given map $\Phi[a,b,c]$, we consider a linear entanglement witness constructed by this map in the form $[\mathbb{I}\otimes \Phi[a,c,b]](\ket{\phi^+_3}\bra{\phi^+_3})$, where $|\phi^+_3\ra=\frac{1}{\sqrt{3}}(|00\ra+|11\ra+|22\ra)$ is the maximally entangled state. For general values of $[a,b,c]$, the form of this linear EW is 
 found to be 
 \begin{equation}\label{linearabc}
 	W[a,b,c]= \frac{1}{3}D[a,b,c]-\ket{\phi^+_3}\bra{\phi^+_3},
 \end{equation}
 where
 \begin{equation}
 	D[a,b,c]=\sum_{i=0}^2 \ket{i}\bra{i}\otimes \big[ (a+1)\ket{i}\bra{i}+ b \ket{i+1}\bra{i+1}+c \ket{i+2}\bra{i+2}\big].
 \end{equation} 
   To construct a nonlinear EW by our method, we take the Schmidt-decomposed entangled state as $|\phi\ra=\sqrt{p_0}|00\ra+\sqrt{p_1}|11\ra+\sqrt{p_2}|22\ra$ which upon taking the envelope leads to the following form of the matrix $\Delta_\Phi[a,b,c](\rho)$:
\begin{equation}
    \begin{aligned}
        (\Delta_\Phi[a,b,c](\rho))_{i,i} &= a \rho_{i,i,i,i} 
        + b\rho_{i,i-1,i,i-1} 
        + c\rho_{i,i-2,i,i-2}, \\
        (\Delta_\Phi[a,b,c](\rho))_{i,j} &= -\Re(\rho_{ii,jj}),\hspace{2cm}(i\neq j)
    \end{aligned}
\end{equation}
The nonlinear criterion thus becomes a violation of the positive semidefiniteness of the following $3 \times 3$ matrix
\begin{equation}\label{nonlineaabc}
    \begin{aligned}
        \Delta_\Phi[a,b,c](\rho) \\
        &= \begin{pmatrix}
            a\rho_{00,00}+c\rho_{01,01}+b\rho_{02,02} & -\Re(\rho_{00,11}) & -\Re(\rho_{00,22}) \\
            -\Re(\rho_{00,11}) & b\rho_{10,10} + a \rho_{11,11} + c \rho_{12,12} & -\Re(\rho_{11,22})\\
            -\Re(\rho_{00,22}) & -\Re(\rho_{11,22}) & c\rho_{20,20} + b \rho_{21,21} + a \rho_{22,22}
        \end{pmatrix}
    \end{aligned}
\end{equation}
We emphasize that the required measurements for both the linear and nonlinear witness can be expressed in terms of generalized Gell-Mann matrices \cite{kimura2003bloch,bertlmann2008bloch}, the measurement setup being the same for both the linear and nonlinear witnesses. 

\noindent We can now compare the detection power of linear and nonlinear witnesses for the maps $\Phi[a,b,c]$. 
To make the comparison more visible, we take the parameters $[a,b,c]$ from (\ref{parametrization}), i.e. from the boundary of the set of optimal maps. Therefore for each value of $\theta$, we have 
a linear witness $W_{\Phi[\theta]}$, and a nonlinear witness $F_{\Phi[\theta]}$. For each value of $\theta$, we compare their detection power numerically. To this end, a set of $ 10^6$ two-qutrit states is generated, belonging to the family
\begin{equation}\label{mixedfamilynumerical}
    \rho_{\psi,p}= p \ket{\psi}\bra{\psi}+\frac{1-p}{9}\mathbb{I}_{9\times 9}.
\end{equation}
where,  $\ket{\psi}$ is a randomly generated entangled state \cite{johansson2012qutip} and $p\in[0,1]$ is also randomly chosen. The detection power of the linear and nonlinear witnesses is compared in Fig. \ref{fig:comparePnCPparameters}.
\begin{figure}[H]
	\centering
	\includegraphics[width=11cm]{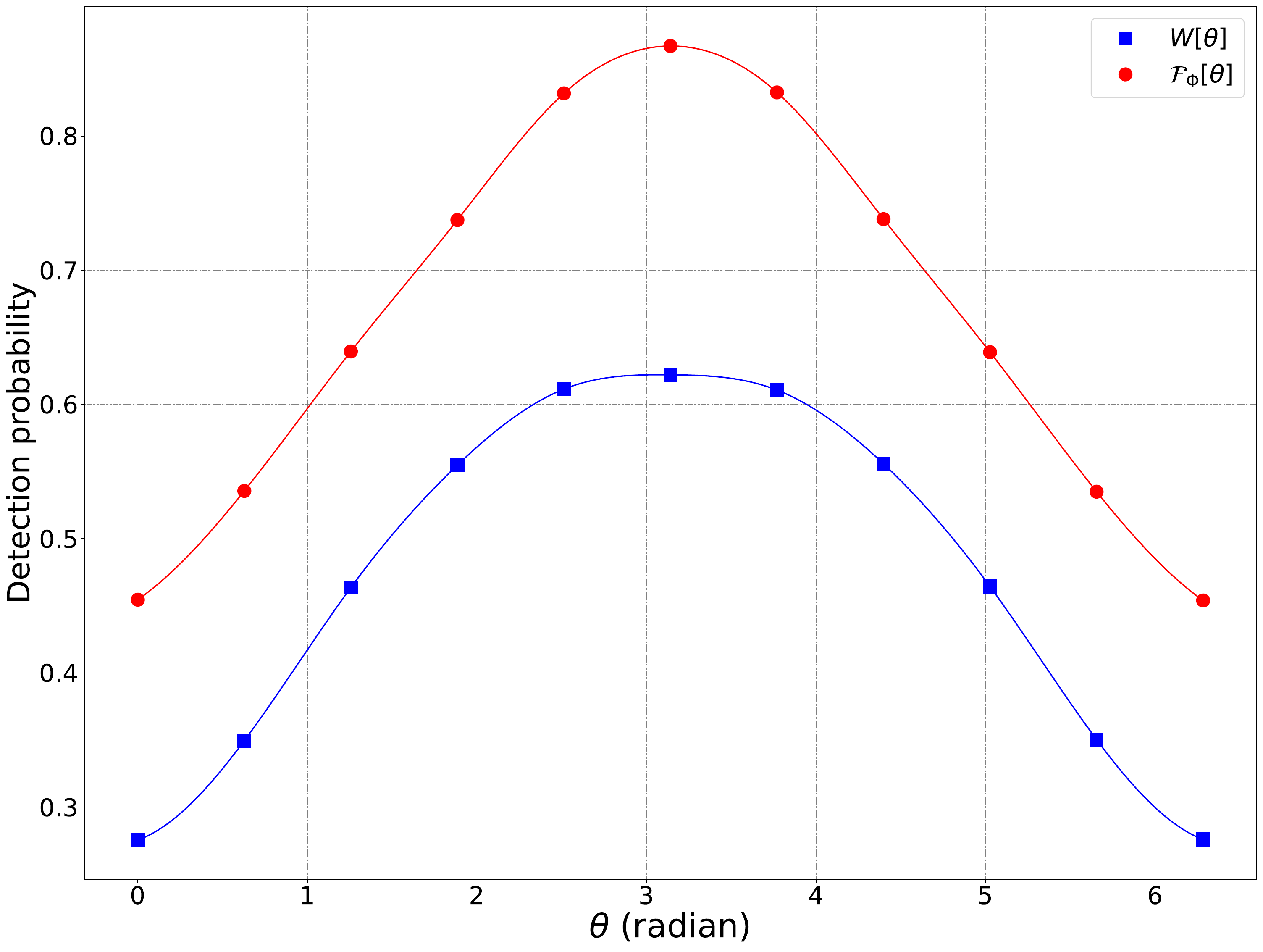}
	\caption{Comparison of the detection power of linear (\ref{linearabc}) and nonlinear witnesses (\ref{nonlineaabc}) for $10^6$ random states of the form (\ref{mixedfamilynumerical}). We filtered the states that satisfy $\tr(W_{\Phi[\theta]}\rho_{\psi,p})\geq \epsilon$ for a positive $\epsilon$ and set  $\epsilon=0.03$ to exclude most of the states that are not detected by either the linear or the nonlinear witnesses.}
	\label{fig:comparePnCPparameters}
\end{figure}

\section{DISCUSSION}
\ni In this work, we have presented a construction of nonlinear entanglement witnesses (EWs) derived as envelopes of families of linear EWs defined over pure bipartite states with fixed Schmidt bases. The resulting nonlinear criteria are formulated as the condition that a certain operator-valued matrix, built from expectation values of the density operator, is positive semidefinite. The positivity of its principal minors then yields a hierarchy of increasingly sensitive detection conditions. This construction effectively bends the hyperplanes associated with linear witnesses into nonlinear hypersurfaces that capture entangled states undetectable by any single linear witness in the generating family.\\

\ni Our method offers several conceptual and practical advantages. First, it is entirely analytical, with no reliance on state-dependent optimization or numerical fitting procedures \cite{Guhne2003,ChenWu2003}. Second, it generalizes the standard theory of linear EWs \cite{Terhal2000,Horodecki1996,lewenstein1998separability} by extending it to nonlinear hypersurfaces while preserving experimental accessibility: all the required expectation values are expressible in terms of a fixed local operator basis, such as Pauli or Gell-Mann matrices, making implementation with current quantum technologies feasible \cite{Barbieri2003}.\\

\ni We demonstrated that these nonlinear witnesses detect strictly more entangled states than any of the linear witnesses from which they are built. In both analytical examples (e.g., mixtures of Bell states) and numerical simulations (e.g., random mixed states), our approach succeeded in identifying entanglement in cases where the linear criteria failed. This enhancement aligns with the geometric intuition that nonlinear surfaces can more closely approximate the true boundary between separable and entangled states \cite{Guhne2002}.\\

\ni Our results suggest several promising directions for future research. One natural extension is to multipartite systems, where the structure of entanglement is more intricate, and conventional linear witnesses often prove inadequate. Another direction is to analyze the robustness of nonlinear EWs under noise and decoherence, which could inform their use in experimental platforms subject to imperfections \cite{Zurek2003}. Finally, it would be valuable to investigate the relation between our nonlinear conditions and those derived from semidefinite programming relaxations \cite{Brandao2005}, to further clarify the power and limitations of our approach.\\

\noindent {{\bf ACKNOWLEDEMENTS}} This research was supported in part by the Iran National Science Foundation, under Grants No. 4022322 and No. QST4040203. The authors wish to thank V. Jannessary, S. Roofeh, A. Farmanian, A. Najafzadeh, F. Rahimi,  F. Kharabat, M.H. Shahbazian and A. Salarpour for their valuable comments and suggestions.\\

\noindent {{\bf DATA AVAILABILITY}} The data that support the findings of this article are not publicly available upon publication because it is not technically feasible and/or the cost of preparing, depositing, and hosting the data would be prohibitive within the terms of this research project. The data are available from the authors upon reasonable request.

\appendix\section*{APPENDIX A: FINDING THE ENVELOPE OF A FAMILY OF PLANES}
In this appendix, we explain in an intuitive manner how the envelope of a set of planes can be found. Consider a convex set $C\subset \mathbb{R}^n$ and a continuous family of hyperplanes $\{H_t\}_{t\in \mathbb{R}}$ such that for each $t$, $C$ lies entirely on one side of $H_t$ (as indicated by the Hahn-Banach theorem). Here we are considering a one-parameter family of planes. The generalization to a multiparameter family is straightforward as we will see. A hyperplane $H_t$ in this family is defined by the linear equation
$$H_t:=\{x\in \mathbb{R}^n\ |\  n_t\cdot x=a_t\},$$
where $n_t$ is a unit normal vector pointing away from $C$. Note that $C$ lies in the half $\{x\in \mathbb{R}^n\ | \ n_t\cdot x \leq a_t\}$. The family being continuous means that the equations of two neighboring planes are given by
\ba
&&n_t\cdot x=a_t\cr
&&n_{t+\epsilon}\cdot x=a_{t+\epsilon},
\ea
which results in the set
\ba
&&n_t\cdot x-a_t=0\cr
&&\frac{d}{dt}(n_{t}\cdot x-a_{t})=0,
\ea
Solution of these two equations which is equivalent to eliminating the parameter $t$ between them, leads to a nonlinear equation for $x$ that is the equation fo the envelope. For an $r$-parameter family of planes the above two equations are replaced with the following $r+1$ equations \cite{nishimura2022hyperplane}
\ba
&&n_{{\bf p}}\cdot x-a_{{\bf p}}=0\cr
&&\frac{\partial}{\partial p_i}(n_{\bf p}\cdot x-a_{{\bf p}})=0,\ \ i=1,\cdots r.
\ea
Solution of these equations, which is equivalent to eliminating the parameters $p_i$, leads to the nonlinear equation of the envelope. In our case, where the parameters $p_i$ are subject to the normalization condition $\sum_i p_i=1$, the derivatives are calculated by using the Lagrange multiplier, hence Eqs. (\ref{lg1}) and (\ref{lg2}) of the main text.

\appendix\section*{APPENDIX B: A CLOSER LOOK AT EXAMPLE OF SEC. \ref{example: Bell states}}
Figure \ref{fig: Bell} in example Sec. \ref{example: Bell states}, shows in a qualitative way why a single linear witness of the form $|\phi\ra\la \phi|^\Gamma$ cannot detect a convex combination of two of the Bell states while a nonlinear witness can. In this appendix, we want to elaborate on this observation and have a closer look at this example. The convex combination of Bell states is of the form 
\be\rho= p_0 |\phi^+\ra\la\phi^+|+p_1 |\psi^+\ra\la\psi^+|+p_2 |\psi^-\ra\la\psi^-|+p_3 |\phi^-\ra\la\phi^-|,\ee
or in matrix form
\be\label{matrix: p}
\rho=\frac{1}{2}\begin{pmatrix}p_0+p_3 &&& p_0-p_3\\
&p_1+p_2 & p_1-p_2&\\ &p_1-p_2 & p_1+p_2&\\p_0-p_3 &&& p_0+p_3
 \end{pmatrix},
\ee
where $p_0+p_1+p_2+p_3=1$. This matrix can be written in terms of three independent parameters, if we expand it in terms of Pauli matrices:
\be
\rho=\frac{1}{4}(\mathbb{I}+x\sigma_x\otimes \sigma_x-y\sigma_y\otimes \sigma_y+z\sigma_z\otimes \sigma_z)
\ee
Note that the minus sign for $y$ is just a matter of convention for later simplicity. In matrix form
\be\label{matrix: x}
\rho=\frac{1}{4}\begin{pmatrix}1+z&&&x+y\\ &1-z&x-y&\\ &x-y&1-z&\\ x+y&&&1+z\end{pmatrix}.
\ee
Comparison of the two expressions (\ref{matrix: p}) and (\ref{matrix: x}) yields the relation between the parameters
\ba
&&x=p_0+p_1-p_2-p_3\cr
&&y=p_0-p_1+p_2-p_3\cr
&&z=p_0-p_1-p_2+p_3.
\ea
The subset of entangled states is obtained by applying the Peres's criterion \cite{peres1996separability} which in this case amounts to demanding one of the eigenvalues of $\rho^\Gamma$ to be zero.  The spectrum of $$\rho^\Gamma=\frac{1}{4}\begin{pmatrix}1+z&&&x-y\\ &1-z&x+y&\\ &x+y&1-z&\\ x-y&&&1+z\end{pmatrix}$$ 
is easily found to be 
\ba
&&\lambda_0=\frac{1}{4}(1-x-y-z)=\frac{1}{2}-p_0\cr
&&\lambda_1=\frac{1}{4}(1-x+y+z)=\frac{1}{2}-p_1\cr
&&\lambda_2=\frac{1}{4}(1+x-y+z)=\frac{1}{2}-p_2\cr
&&\lambda_3=\frac{1}{4}(1+x+y-z)=\frac{1}{2}-p_3.
\ea
The subset of separable states is specified by the condition $\lambda_i\geq 0\ \ \forall i$, which defines an octahedron inside the tetrahedron, shown by a dark color in Fig. \ref{fig: tetrahedron}. The region outside this octahedron is the subset of entangled state; that is the region where one of the parameters $p_i>\frac{1}{2}$. This is the region where the contribution of one of the Bell states is dominant in the mixture of all four Bell states. Consider now  optimal linear witnesses of the form $W_i=|\phi^i\ra\la \phi^i|^\Gamma$, where $|\phi^i\ra$ is one of the Bell states. A simple calculation shows that 
\begin{align*}
&\tr(W_0\rho)\equiv \tr(|\psi^-\ra\la \psi^-|^\Gamma \rho)=\lambda_0\cr\\
&\tr(W_1\rho)\equiv \tr(|\phi^-\ra\la \phi^-|^\Gamma \rho)=\lambda_1\cr\\
&\tr(W_2\rho)\equiv \tr(|\phi^+\ra\la \phi^+|^\Gamma \rho)=\lambda_2\cr\\
&\tr(W_3\rho)\equiv \tr(|\psi^+\ra\la \psi^+|^\Gamma \rho)=\lambda_3.
\end{align*}
This shows that these linear witnesses correspond to the four wedges of the tetrahedron which encompass entangled states, as shown in Fig. \ref{fig: tetrahedron}. Moreover it shows that these witnesses are optimal and touch four of the faces of the octahedron. Referring to Fig. \ref{fig: tetrahedron}, it is seen that none of these witnesses are capable to detect more than one Bell state, although they can detect a certain convex combination of Bell states (depending on which wedge they correspond to). In contrast, we can now consider a nonlinear witness like 
\begin{equation}\label{nonlinear witness for appendix}
	F^{(1)}_T(\rho) = \rho_{00,00} \rho_{11,11}- (\Re{\rho_{01,10}})^2= \frac{1}{16}\Big((1+z)^2-(x-y)^2\Big)
\end{equation}
from which we see that the surface $F^{(1)}_T(\rho)=0$ corresponds to the combination of two planes, namely $\frac{1}{16}(1+z+ x-y)(1+z-x+y)=0$, or $\lambda_1\lambda_2=0$ which touch the octahedron of PPT states.   Similarly for a nonlinear witness of the form 
$F_T^{(2)}(\rho)=\rho_{01,01}\rho_{10,10}-(\Re{\rho_{00,11}})^2$, we find that $F^{(2)}(\rho)=0$ is equivalent to $\frac{1}{16}(1+x+y-z)(1-x-y-z)=0$ or $\lambda_0\lambda_3=0$, which corresponds to the combination of two other planes, touching the octahedron. 
\begin{figure}[H]
	\centering
	\includegraphics[width=15cm]{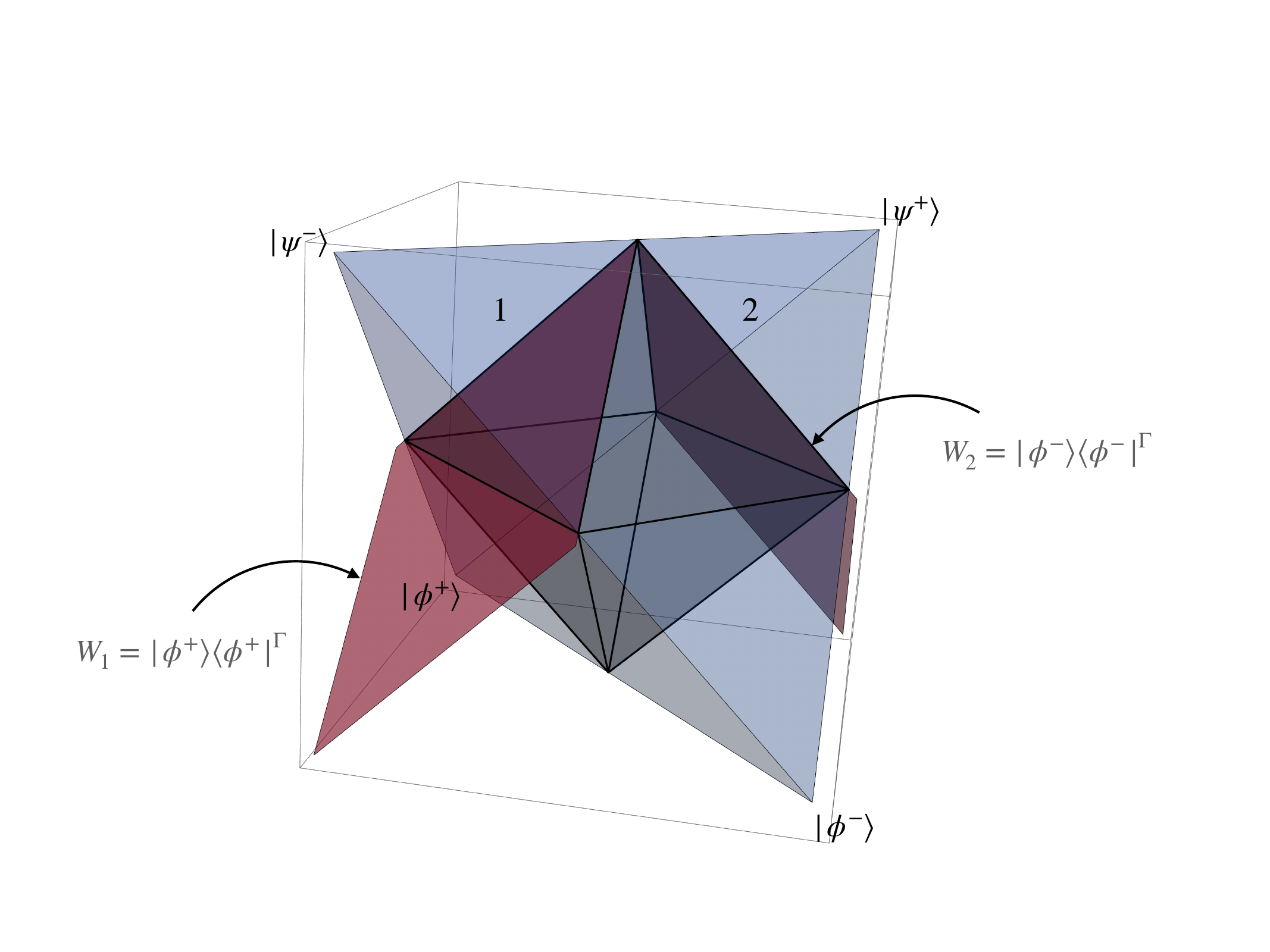}\vspace{-1cm}
\caption{The linear witness $W_1=\ket{\phi^+}\bra{\phi^+}^\G=\frac{1}{2}\mathbb{I}-\ket{\psi^-}\bra{\psi^-}$ can detect the entanglement of the state $|\psi^-\ra\la \psi^-|$ and all  the states in the region No. 1 of the tetrahedron.  Similarly the  witness $W_2=\ket{\phi^-}\bra{\phi^-}^\G = \frac{1}{2}\mathbb{I}-\ket{\psi^+}\bra{\psi^+}$ can detect the entanglement of the state $|\psi^+\ra\la \psi^+|$ and all the states in the region No. 2 in the tetrahedron. The nonlinear witness (\ref{nonlinear witness for appendix}) can detect all the states in both the regions No. 1 and No. 2.}
	\label{fig: tetrahedron}
\end{figure} 

	\newpage
	\bibliography{refs.bib}

\end{document}